\documentclass[10pt, doublecolumn]{IEEEtran}
\usepackage{epsfig,latexsym}
\usepackage{float}
\usepackage{indentfirst}
\usepackage{amsmath}
\usepackage{bm}
\usepackage{amssymb}
\usepackage{times}

\usepackage{algorithm}
\usepackage[noend]{algpseudocode}
\usepackage{subfigure}
\usepackage{psfrag}
\usepackage{hyperref}
\usepackage{cite}
\usepackage{lastpage}
\usepackage{fancyhdr}
\usepackage{color}
 \usepackage{amsthm}
\usepackage{bigints}
\newtheorem{Lemma}{Lemma}
\newtheorem{lemma}[Lemma]{$\mathbf{Lemma}$}

\newcounter{problem}
\newcounter{save@equation}
\newcounter{save@problem}
\makeatletter
\newenvironment{problem}
{\setcounter{problem}{\value{save@problem}}%
  \setcounter{save@equation}{\value{equation}}%
  \let\c@equation\c@problem
  \subequations
}
{\endsubequations
  \setcounter{save@problem}{\value{equation}}%
  \setcounter{equation}{\value{save@equation}}%
}

\begin{document}
\title{\vspace{-0.5em} \huge{ Potentials and Limits of Using Preconfigured Spatial Beams as Bandwidth Resources: Beam Selection vs Beam Aggregation  }}

\author{ Zhiguo Ding, \IEEEmembership{Fellow, IEEE}     \thanks{ 
  
\vspace{-2em}

    Z. Ding  
 is    with the School of
Electrical and Electronic Engineering, the University of Manchester, Manchester, UK (email: \href{mailto:zhiguo.ding@manchester.ac.uk}{zhiguo.ding@manchester.ac.uk}).

  }\vspace{-2.5em}}
 \maketitle

\vspace{-1em}
\begin{abstract}
This letter studies  how to   use    spatial beams preconfigured in  a legacy spatial division multiple access (SDMA) network   as bandwidth resources via the implementation of non-orthogonal multiple access (NOMA). Two different beam management schemes, namely beam selection and beam aggregation, are developed to improve the overall system throughput   without consuming extra spectrum or changing the performance of the legacy network. Analytical and simulation results are presented to show that the two schemes realize different tradeoffs between system performance and complexity.  
\end{abstract}\vspace{-0.1em}

\begin{IEEEkeywords}
NOMA, SDMA,  beam management, power allocation, outage probability. 
\end{IEEEkeywords}
\vspace{-1em} 
\section{Introduction}
Recently, using  non-orthogonal multiple access (NOMA) as a type of add-ons in communication networks  has received a lot of attention, because it makes the implementation of NOMA more flexible in practice and also improves   the performance of various  communication networks   in a spectrally efficient manner \cite{9693536, 9693417}. For example, consider   a  legacy spatial division multiple access (SDMA) network, where  spatial beams have already been configured to serve legacy primary users \cite{1261332}. Conventional   multiple-input multiple-output NOMA (MIMO-NOMA) approaches need to redesign the legacy network in order to serve new secondary users  \cite{8352619,9779941,9787502}. An  alternative  is to directly use those preconfigured beams for serving the new   users, which makes the implementation of NOMA and the admission of the new users transparent to  the legacy network. As shown in \cite{bacnomamtc1}, these preconfigured spatial beams indeed have the potential to be used as  bandwidth resources, similar to orthogonal frequency-division multiplexing (OFDM) subcarriers, but the benefits of using these beams can be severely  limited by inter-beam interference. 
  
The aim of this letter is to provide more detailed analysis for the potentials and limits of using  these preconfigured spatial beams, by focusing on a simple case with multiple preconfigured beams and a single secondary user.   For the scenario where the secondary user's signals on different beams are independently encoded, \cite{bacnomamtc1} shows that the use of beam selection, i.e., selecting a single beam to serve the secondary user, is  optimal, and the outage probability achieved by   beam selection is analyzed in this letter. Despite its simplicity,     beam selection   suffers a drawback that its performance is degraded with  more beams available, because of  inter-beam interference. To mitigate this inter-beam interference, a new scheme, termed beam aggregation, is proposed in the letter, where the secondary user is encouraged to use multiple beams and its signals on these beams are jointly  designed.   Simulation results are presented to show that beam aggregation can effectively suppress  inter-beam interference and always outperform beam selection, but at a price of more system complexity.  
 \section{System Model} 
 Consider a legacy SDMA network with one $N$-antenna base station and $M$ single-antenna primary users, denoted by ${\rm U}^P_m$, whose channel vectors are denoted by $\mathbf{g}_m$.  Assume that $M$ beamforming vectors, denoted by $\mathbf{f}_m$, have been configured for the primary users. Furthermore, assume that  zero-forcing beamforming is used, i.e.,   $\begin{bmatrix}\mathbf{f}_1&\cdots &\mathbf{f}_M \end{bmatrix}= \mathbf{G}(\mathbf{G}^H\mathbf{G})^{-1}\mathbf{D}$, where $\mathbf{G}=\begin{bmatrix}\mathbf{g}_1 &\cdots &\mathbf{g}_M \end{bmatrix}$, $\mathbf{D}_{i,i} =\left( M[(\mathbf{G}^H\mathbf{G})^{-1}]_{i,i} \right)^{-\frac{1}{2}}$ and $\mathbf{A}_{i,j}$ denotes an element of $\mathbf{A}$ on its $i$-th row and $j$-th column \cite{1261332}. As a result,   $\mathbf{g}_m^H\mathbf{f}_i=0$, for $m\neq i$. 

As shown in \cite{bacnomamtc1}, these preconfigured beams, $\mathbf{f}_m$, can be exploited as a type of bandwidth resources and used to serve additional secondary users, similarly to OFDM subcarriers. In this paper, a simple case with a single secondary user, denoted by ${\rm U}^S$,  is focused. 
 On each of the $M$ beams, the base station superimposes ${\rm U}^P_m$'s signal, denoted by $s_m^P$, with  the secondary user's signal,  denoted by $s_m^S$, where the transmit powers of  $s_m^P$ and  $s_m^S$ are denoted by     $\alpha_m^P\rho$ and  $\alpha_m^S\rho$, respectively,   $\rho$ denotes the transmit power budget, $\alpha^P_m$ and $\alpha^S_m$ denote the power allocation coefficients,  and   $\alpha^P_m+\alpha^S_m\leq 1$. Therefore, the system model at ${\rm U}^S$ is given by
 \begin{align}\label{yS}
y^S = \mathbf{h}^H \sum^{M}_{i=1} \mathbf{f}_i\left(\sqrt{\alpha^P_i\rho}s_i^P+\sqrt{\alpha^S_i\rho}s_i^S\right) +n^S,
\end{align}
where $\mathbf{h}$ denotes the secondary user's channel vector and $n^S$ denotes the  Gaussian noise with its power normalized.

Assuming that  $s_m^S$, $1\leq m \leq M$, are independently encoded,  on beam $\mathbf{f}_m$, the secondary user ${\rm U}^S$ can decode $s_m^P$  with the following data rate: 
\begin{align} 
\tilde{R}_m = \log
\left(1+\frac{h_m\alpha^P_m}{h_m\alpha^S_m +\underset{i\neq m}{\sum}h_i\left(\alpha^P_i+\alpha^S_i\right)+\frac{1}{\rho}}
\right),
\end{align}
where $h_m=| \mathbf{h}^H\mathbf{f}_m|^2$. Assume that all the primary users have the same target data rate, denoted by $R^P$.   If $\tilde{R}_m\geq R^P$, ${\rm U}^S$ can carry out successive interference cancellation (SIC) successfully  and decode $s_m^S$ with the following data rate:
\begin{align}
 {R}_m = \log
\left(1+\frac{h_m\alpha^S_m}{ \underset{i\neq m}{\sum}h_i\left(\alpha^P_i+\alpha^S_i\right)+\frac{1}{\rho}}
\right).
\end{align}

As shown in \cite{bacnomamtc1}, with $s_m^S$, $1\leq m \leq M$, independently encoded, the optimal resource allocation strategy is to select a single beam to serve the secondary user, and the performance of   beam selection is analyzed in the following section. 
\vspace{-1em}
\section{Beam Selection }
The performance analysis of beam selection requires the explicit expressions for the power allocation coefficients, $\alpha^P_m$ and $\alpha^S_m$.  Depending on whether ${\rm U}^S$ is active on the beams, they  can be expressed differently.  
For the case that ${\rm U}^S$ is not active on beam $\mathbf{f}_m$, denote the corresponding power allocation coefficients by $\alpha_{m,I}^P$ and $\alpha_{m,I}^S$. Otherwise, they are denoted by $\alpha_{m,II}^P$ and $\alpha_{m,II}^S$.

If ${\rm U}^S$ is not active on beam $\mathbf{f}_m$, $\alpha_{m,I}^S=0$, and hence  the use of $\alpha_{m,I}^P=\min\left\{1, \frac{\epsilon_P}{\rho g_m} \right\}$ is sufficient to ensure that ${\rm U}^P_m$ decodes $s_m^P$, i.e., $\log(1+\rho \alpha_{m,I}^P g_m)\geq  {R}^P$, where $g_m=|\mathbf{g}_m^H\mathbf{f}_m|^2$ and $\epsilon_P=2^{ {R}^P}-1$. 

If ${\rm U}^S$ is active on beam $\mathbf{f}_m$,   ${\rm U}^P_m$'s data rate on beam $\mathbf{f}_m$ is given by
$
 {R}_m^P = \log
\left(1+\frac{g_m\alpha^P_{m,II}}{g_m\alpha^S_{m,II}  +\frac{1}{\rho}}
\right)$. 
To ensure that $ {R}_m^P\geq   {R}^P$, $\alpha_m^S$ should be chosen as follows:
\begin{align}
\alpha^S_{m,II} \leq \max\left\{ 0, \frac{g_m- \frac{\epsilon_P}{\rho}}{(\epsilon_P +1)g_m}\right\},
\end{align}
and $\alpha^P_{m,II}=1-\alpha^S_{m,II}$. To ensure  that ${\rm U}^S$ decodes  $s_m^P$,  $\tilde{R}_{m} \geq  
\bar{R}_m^P$ is required, which leads to the following constraint: 
\begin{align}
\alpha^S_{m,II} \leq  \max\left\{ 0, \frac{h_m  -\epsilon_P \underset{i\neq m}{\sum}h_i \alpha^P_{i,I} -\frac{\epsilon_P }{\rho}}{(1+\epsilon_P)h_m}\right\},
\end{align}
where $\alpha^P_{i,I} $ is used instead of $\alpha^P_{i,II} $ because the secondary user is not active on beam $\mathbf{f}_i$, $i\neq m$. Therefore, $\alpha^S_{m,II}$ can be chosen as follows: 
\begin{align}\label{alphaII}
\alpha^S_{m,II} =&\min\left\{ \max\left\{ 0, \frac{ g_m- \frac{\epsilon_P}{\rho}}{(\epsilon_P +1)g_m}\right\}\right.
\\\nonumber&
\left.\max\left\{ 0,\frac{h_m  -\epsilon_P \underset{i\neq m}{\sum}h_i \alpha^P_{i,I} -\frac{\epsilon_P }{\rho}}{(1+\epsilon_P) h_m}\right\}\right\}.
\end{align}
 
Assume that beam $\mathbf{f}_{m^*}$ is selected, which means that the outage probability of interest can be expressed as follows: ${\rm P}^o\triangleq 1 - {\rm P}\left(  \tilde{R}_{m^*}\geq R^P,  {R}_{m^*}\geq R^S, \right)$, where $R^S$ denotes the secondary user's target data rate. Analyzing ${\rm P}^o$ is challenging since $h_m$, $g_m$ and $h_i$, $m\neq i$, are correlated. The following lemma shows an interesting property of the NOMA transmission protocol.  
\begin{lemma}\label{lemma1}
Assume that the users' channels are independent and identically complex Gaussian distributed with zero mean and unit variance. ${\rm P}^o\rightarrow 0$ for $\rho \rightarrow \infty$. 
\end{lemma}
\begin{proof}
See Appendix A. 
\end{proof}

{\it Remark 1:}  Lemma \ref{lemma1} shows that there is no error floor for the considered outage probability, despite strong co-channel interference. Or in other words, by using NOMA,  an additional user can be supported  with an arbitrarily high data rate,  without consuming extra spectrum but simply   increasing the transmit power.  This property is   valuable   given the scarcity of the spectrum available to communications.  

{\it Remark 2:} Our conducted simulation results show that increasing $M$ degrades the outage probability, which can be explained in the following. Each beam can be viewed as an OFDM subcarrier. By increasing $M$ is similar to increase the number of subcarriers. If the overall bandwidth is kept the same, increasing the number of subcarriers reduces the bandwidth of each subcarrier, and hence degrades the user's performance if a user can use  one subcarrier only. Similarly, for the considered NOMA system, increasing $M$ increases inter-beam interference, which causes the performance degradation. In the next section, two schemes based beam aggregation  is introduced to mitigate this inter-beam interference. 

\section{Beam Aggregation} 
Unlike beam selection,   beam aggregation encourages   the secondary user  to use  multiple beams. Denote $\mathcal{D}$ by the set including the beams   used by the secondary user,  which means that   the system model in \eqref{yS} can be rewritten as follows: 
\begin{align}\label{ys2}
y^S =& \mathbf{h}^H \sum_{i\in\mathcal{D}} \mathbf{f}_i\left(\sqrt{\alpha^P_i\rho}s_i^P+\sqrt{ \alpha^S_i\rho} \beta_is^S\right) 
\\\nonumber &+ \mathbf{h}^H \sum_{j\in\mathcal{D}^c}  \mathbf{f}_j\sqrt{\alpha^P_j\rho}s_j^P +n^S,
\end{align}
where $s^S$ denotes the secondary user's signal, $\beta_i=\frac{( \mathbf{h}^H \mathbf{f}_i)^*}{\sqrt{| \mathbf{h}^H \mathbf{f}_i|^2}} $ is introduced in order to ensure coherent  combining and also avoid changing  the overall transmit power, and   $\mathcal{D}^c$ denotes the complementary set of $\mathcal{D}$.

\vspace{-1em}
\subsection{Scheme I}
A low-complexity and straightforward scheme is to ask the secondary user to directly decode its own message by treating the primary users' signals as noise.  By using \eqref{ys2}, it is straightforward to show that the following data rate is achievable to the secondary user:
\begin{align}
R_S = \log\left(1+  \frac{\left( \underset{i\in \mathcal{D}}{\sum}   \sqrt{ h_i\alpha^S_i } 
\right)^2}{ \sum_{j=1}^Mh_j  {\alpha^P_j } +\frac{1}{\rho}  } 
\right).
\end{align} 

 To decide the power allocation coefficients, $\alpha_m^P$ and $\alpha_m^P$, recall that the system model at the primary users can be expressed as follows: 
\begin{align}\nonumber
y_m^P=& \mathbf{g}^H_m   \mathbf{f}_m\left(\sqrt{\alpha^P_m\rho}s_m^P+ \sqrt{\alpha^S_m\rho }\beta_m s^S\right) +n_m^P, \quad m\in \mathcal{D}
\\  y_m^P=& \mathbf{g}_m^H    \mathbf{f}_m\sqrt{\alpha^P_m\rho}s_m^P +n_m^P, \quad m\in \mathcal{D}^c,
\end{align}
which means the following data rates  are achievable
\begin{align}
R_m^P=& \log\left(1+\frac{g_m{\alpha^P_m } }{g_m\alpha^S_m\beta_m^2  +\frac{1}{\rho}}\right), \quad m\in \mathcal{D}
\\\nonumber R_m^P=& \log\left(1+g_m{\alpha^P_m\rho}  \right), \quad m\in \mathcal{D}^c,
\end{align}
where   $n_m^P$ denotes the noise at ${\rm U}^P_m$. 

In order to ensure that $R_m^P\geq R^P$, the following choices for the power allocation coefficients can be used:
\begin{align}\nonumber
\alpha^S_{m} =\max\left\{ 0, \frac{g_m- \frac{\epsilon_P}{\rho}}{(\epsilon_P +1)g_m}\right\}, 
\alpha^P_{m} =\min
\left\{1, \frac{\epsilon_P\left(g_m +\frac{1}{\rho}\right)}{g_m (1+\epsilon_P)}\right\},
\end{align}
for $m\in \mathcal{D}$. Otherwise, $\alpha^P_{m}=\min\left\{1, \frac{\epsilon_P}{\rho g_m} \right\}$ and $\alpha^S_{m} =0$.

As shown in the next section, the performance of Scheme I can ensure that the secondary user's data rate grows by increasing $M$; however, it suffers some performance loss compared to beam selection, particularly at high SNR. 
\vspace{-1em}
\subsection{Scheme II} 
To facilitate the description of the scheme, assume that the secondary user's channel gains are ordered as follows: $h_1\geq \cdots \geq h_M$, which means that a reasonable choice for $\mathcal{D}$ is given by $\mathcal{D}=\{1, \cdots, |\mathcal{D}|\}$, where $|\mathcal{D}|$ denotes the size of $\mathcal{D}$. Scheme II is to carry out SIC on the beams included in $\mathcal{D}$ successively, i.e., $s_i^P$ is decoded and removed before decoding $s_j^P$, for $i<j$ and $i,j\in\mathcal{D}$. After $s_m^P$, $m\in \mathcal{D}$, are decoded, the secondary user can decode its own signal, $s^S$. 

Therefore, on beam $\mathbf{f}_m$, $m\in \mathcal{D}$, the   data rate    for the secondary user to decode $s_m^P$ is given by 
\begin{align}\nonumber
\tilde{R}_m = \log
\left(1+\frac{h_m\alpha^P_m}{  \sum^{M}_{j=m+1}  h_j \alpha^P_j     +|\underset{i\in \mathcal{D}}{\sum} \sqrt{h_i\alpha^S_i }|^2+\frac{1}{\rho}}
\right).
\end{align}
Conditioned on $\tilde{R}_m\geq R^P$, $\forall m\in\mathcal{D}$, the secondary user can carry out SIC successfully  and realize the following data rate for its own signal:
\begin{align}
 {R}^S = \log
\left(1+\frac{| \underset{i\in\mathcal{D}}{\sum } \sqrt{h_i\alpha^S_i}    |^2  }{ \underset{j\in \mathcal{D}^c}{\sum}h_j \alpha^P_j+\frac{1}{\rho}}
\right).
\end{align}

Compared to Scheme I,  for Scheme II,  the choices for   the power allocation coefficients, $\alpha_m^P$ and $\alpha_m^S$, $m\in\mathcal{D}$, are more difficult to find, because they need to ensure that $s_m^P$, $m\in \mathcal{D}$, can be decoded by both ${\rm U}_m^P$ and ${\rm U}^S$. Compared to beam selection which uses a single beam only, the design of Scheme II is also more challenging since    multiple beams are used and hence   a choice of $\alpha_m^S$ affects that of $\alpha_i^S$, $m\neq i$. Note that designing    the power allocation coefficients to maximize the the secondary user's data rate     can be formulated as the following optimization problem:  
 \begin{problem}\label{pb:1} 
  \begin{alignat}{2}
\underset{\alpha_m^P,\alpha_m^S}{\rm{max}} &\quad    
 \log
\left(1+\frac{| \underset{i\in\mathcal{D}}{\sum }  \sqrt{h_i\alpha^S_i}   |^2  }{ \underset{j\in \mathcal{D}^c}{\sum}h_j\alpha^P_j+\frac{1}{\rho}}
\right)\label{1tobj:1} \\ \label{1tst:1}
\rm{s.t.} & \quad  \tilde{R}_m  \geq R^P_m, \forall m\in \mathcal{D}
\\
& \quad     \alpha_m^P\geq \frac{\epsilon_P\left(g_m +\frac{1}{\rho}\right)}{g_m (1+\epsilon_P)},\forall m\in \mathcal{D} \label{1tst:2}
\\
& \quad     \alpha_m^S+\alpha_m^P\leq 1,\forall m\in \mathcal{D}  \label{1tst:3}
\\
& \quad     \alpha_m^S\geq 0,  \alpha_m^P\geq 0,\forall m\in \mathcal{D}  \label{1tst:4}. 
  \end{alignat}
\end{problem} 
Note that constraint \eqref{1tst:1} ensures that ${\rm U}^S$ can decode $s_m^P$, and constraint \eqref{1tst:2} ensures that serving ${\rm U}^S$ on beam $\mathbf{f}_m$ does not degrade ${\rm U}_m^P$'s performance, $m\in \mathcal{D}$. 

It is straightforward to show that problem \ref{pb:1} is not concave, but it can be recast to an equivalent concave form as follows. First, note that  the following power allocation coefficients are  fixed: $\alpha^P_{m}=\min\left\{1, \frac{\epsilon_P}{\rho g_m} \right\}$ and $\alpha^S_{m} =0$,   $m\in \mathcal{D}^c$, as shown in the previous section. By using the fact that $\alpha^P_{m}$ and $\alpha^S_{m}$, $m\in \mathcal{D}^c$, are constants, problem \ref{pb:1} can be equivalently recast as follows:
 \begin{problem}\label{pb:2} 
  \begin{alignat}{2}
\underset{\alpha_m^P,\alpha_m^S}{\rm{max}} &\quad    
 \sum^{|\mathcal{D}|}_{i=1} \sqrt{h_i \alpha^S_i}     \label{2tobj:1} \\\nonumber
\rm{s.t.} & \quad  
 \frac{h_m\alpha^P_m}{ \tau_D + \sum_{j=m+1}^{|\mathcal{D}|} h_j \alpha^P_j     + \left(\sum^{|\mathcal{D}|}_{i=1} \sqrt{h_i \alpha^S_i}  \right)^2   }
 \geq \epsilon_P,\\&\quad\quad\quad\quad\quad\quad\quad\quad\quad\quad\quad\quad\quad 1\leq m\leq |\mathcal{D}| \label{2tst:1}
\\
& \quad     \alpha_m^P\geq \eta_m,\forall m\in \mathcal{D} \label{2tst:2}
\\
& \quad      \eqref{1tst:3},   \eqref{1tst:4}\nonumber
  \end{alignat}
\end{problem} 
where   $\eta_m=\frac{\epsilon_P\left(g_m+\frac{1}{\rho}\right)}{g_m (1+\epsilon_P)}$ and $\tau_D=\sum_{j=|\mathcal{D}|+1}^{M} h_j \alpha^P_j +\frac{1}{\rho}$. 
Problem \ref{pb:2} can  be further equivalently recast  as follows:
 \begin{problem}\label{pb:3} 
  \begin{alignat}{2}
\underset{\alpha_m^P,\alpha_m^S}{\rm{max}} &\quad    
 \sum^{|\mathcal{D}|}_{i= 1}  \sqrt{h_i \alpha^S_i}   \label{3tobj:1} \\
\rm{s.t.} & \quad  \nonumber
    \epsilon_P\sum_{j=m+1}^{|\mathcal{D}|} h_j \alpha^P_j     +\epsilon_P \left( \sum^{|\mathcal{D}|}_{i=1} \sqrt{h_i \alpha^S_i}  \right)^2  \\&\quad\quad\quad - h_m\alpha^P_m  + \epsilon_P \tau_D\leq 0, 1\leq m\leq |\mathcal{D}| \label{3tst:1}
\\
& \quad   \eqref{2tst:1},    \eqref{1tst:3},   \eqref{1tst:4}. \nonumber
  \end{alignat}
\end{problem} 
Problem \ref{pb:3} is not a concave problem mainly due to the fact that $ \left( \sum^{|\mathcal{D}|}_{i=1} \sqrt{h_i \alpha^S_i}  \right)^2 $ is a concave function of $\alpha_i^S$. Fortunately, problem \ref{pb:3} can be still recast to a concave problem as shown in the following. Define $x_m=\sqrt{\alpha_m^S}$, and problem \ref{pb:3} can be reformulated  as follows:
 \begin{problem}\label{pb:4} 
  \begin{alignat}{2}
\underset{\alpha_m^P,x_m}{\rm{max}} &\quad    
 \sum^{|\mathcal{D}|}_{i= 1} \sqrt{h_i}x_i   \label{4tobj:1} \\
\rm{s.t.} & \quad  
    \epsilon_P\sum_{j=m+1}^{|\mathcal{D}|} h_j \alpha^P_j    +\epsilon_P \left( \sum^{|\mathcal{D}|}_{i=1} \sqrt{h_i } x_i \right)^2  \nonumber\\ &\quad - h_m\alpha^P_m   + \epsilon_P \tau_D\leq 0, 1\leq m\leq |\mathcal{D}| \label{4tst:1}
\\
& \quad   \alpha_m^P\geq \eta_m, 1\leq m\leq |\mathcal{D}| \label{4tst:2}
\\
& \quad    x_m^2+\alpha_m^P\leq 1, 1\leq m\leq |\mathcal{D}| \label{4tst:3}
\\
& \quad    x_m\geq 0, \alpha_m^P\geq 0, 1\leq m\leq |\mathcal{D}| .\label{4tst:4}
  \end{alignat}
\end{problem} 
Note that $ \left( \sum^{|\mathcal{D}|}_{i=1} \sqrt{h_i } x_i \right)^2$ is a convex function, which means that the constraint function in \eqref{4tst:1} is convex. Similarly, constraint \eqref{4tst:3} is also convex, and all the other functions in problem \ref{pb:4} are affine functions, which means that problem \ref{pb:4} is a concave problem and can be easily solved by using those optimization solvers \cite{Boyd}. 

{\it Remark 3:} The performance of beam aggregation is depending on the choice of $\mathcal{D}$. For the simulations conducted for this paper, an exhaustive search is used to find the optimal choice of $\mathcal{D}$, where designing a more computationally efficient way to optimize $\mathcal{D}$ is an important direction for future research.  
 
 \section{ Simulation Results}
 In this section, the performance of the considered beam management    schemes is investigated and compared  by using computer simulation results. 
 
In Fig. \ref{fig1}, the performance of beam selection based NOMA transmission is studied by using two metrics, namely the outage probability and ergodic data rates. In particular, Fig. \ref{fig1a}  demonstrates that there is no outage probability error floor for the beam selection scheme, which confirms Lemma \ref{lemma1}. In addition, Fig. \ref{fig1a} shows that the performance of beam selection   is degraded by increasing $M$, as discussed in Remark 2. Fig. \ref{fig1b} shows that the use of NOMA transmission yields a considerably large ergodic data rate for the secondary user, and it is important to point  out that such a large data rate is supported without consuming extra spectrum or degrading the performance of the legacy network.

In Fig. \ref{fig2}, the performance of beam aggregation based NOMA transmission is studied by using beam selection as the benchmarking scheme. Fig. \ref{fig2a} shows that the use of the first type of beam aggregation can result in a moderate performance gain over beam selection, particularly with a large $M$ and at low SNR. Fig. \ref{fig2b} demonstrates that  the second type of beam aggregation can always outperform beam selection and the  beam aggregation scheme I. More importantly, the use of the second type of beam aggregation can ensure robust performance regardless of the choices of the beam number. But it is important to point out that the use of the   second type of beam aggregation results in more system complexity. 
 
 \begin{figure}[t] \vspace{-1em}
\begin{center}
\subfigure[Outage Performance $R^S=1 BPCU$]{\label{fig1a}\includegraphics[width=0.35\textwidth]{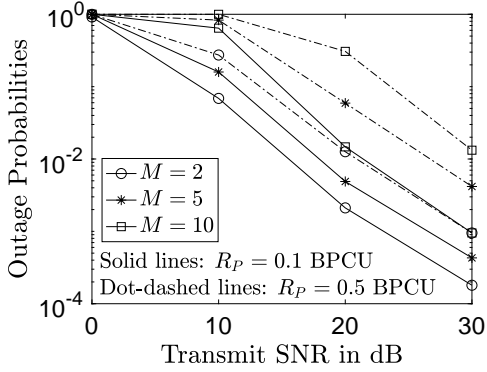}}\hspace{2em}
\subfigure[Ergodic Data Rates ]{\label{fig1b}\includegraphics[width=0.35\textwidth]{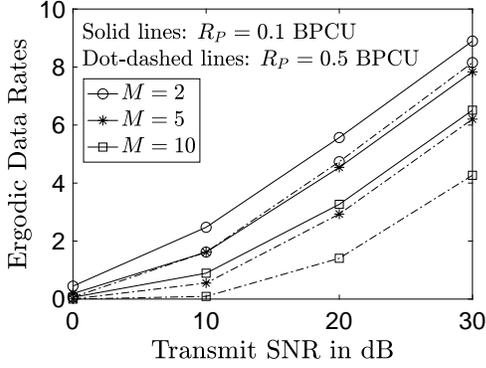}} \vspace{-1.5em}
\end{center}
\caption{ Performance of beam selection based NOMA transmission. $N=M$ and BPCU denotes bits per channel use.   \vspace{-1em} }\label{fig1}\vspace{-1em}
\end{figure}

 \begin{figure}[t] \vspace{-1em}
\begin{center}
\subfigure[Beam Selection vs Beam Aggregation Scheme I]{\label{fig2a}\includegraphics[width=0.35\textwidth]{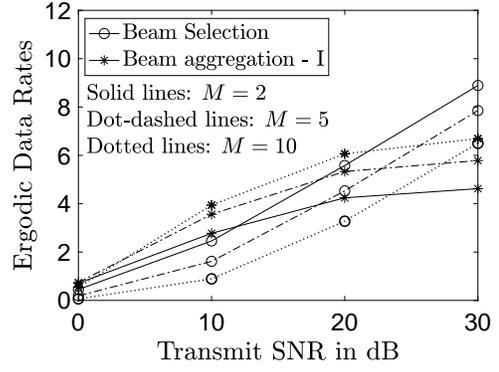}}\hspace{2em}
\subfigure[Beam Selection vs Beam Aggregation Scheme II]{\label{fig2b}\includegraphics[width=0.35\textwidth]{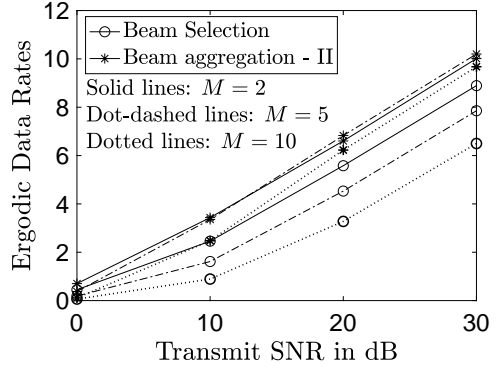}} \vspace{-1.5em}
\end{center}
\caption{ Comparison between the beam selection and beam aggregation based schemes. $N=M$ and $R^P=0.1$ BPCU.     \vspace{-1em} }\label{fig2}\vspace{-1em}
\end{figure}

\vspace{-1em}
\section{Conclusions}
This  letter has demonstrated   how those   spatial beams preconfigured in  a legacy SDMA network can be used  as bandwidth resources via the implementation of NOMA. Two different NOMA schemes have been developed to ensure an additional secondary user served without consuming extra spectrum or changing the performance of the legacy network. Analytical and simulation results have been presented to show that the two schemes realize different tradeoffs between system performance and complexity.  

\vspace{-0.5em}
 \appendices
 \section{Proof for Lemma \ref{lemma1}}

To facilitate the performance analysis, the outage probability is rewritten as follows: 
\begin{align}
{\rm P}^o =& {\rm P}\left(  \alpha^S_{m,II} = 0, \forall m \right)\\\nonumber &+\underset{\mathcal{S}_i\subseteq \mathcal{S}}{\sum}    {\rm P}\left(\alpha^S_{m,II} >0, \forall m\in \mathcal{S}_i ,\alpha^S_{j,II} =0, \forall j\in \mathcal{S}_i^c  \right.\\\nonumber &\left.\max \left\{  \gamma_m
, m\in \mathcal{S}_i   \right\} \leq \epsilon_s,   \right),
\end{align}
where $\gamma_m=\frac{h_m\alpha^S_{m,II}}{ \underset{i\neq m}{\sum}h_i\alpha^P_{i,I}+\frac{1}{\rho}}$, $\mathcal{S}=\{1, \cdots, M\}$, $\mathcal{S}_i$  is a subset of $\mathcal{S}$, and  $\mathcal{S}_i^c$ is a complementary set of $\mathcal{S}_i$. 

${\rm P}^o $ can be first upper bounded as follows:
\begin{align}\label{eq7xx}
{\rm P}^o \leq &  
  {\rm P}\left(  \alpha^S_{1,II} = 0 \right) +\underset{\mathcal{S}_i\subset \mathcal{S}}{\sum}    {\rm P}\left( \alpha^S_{j,II} =0, \forall j\in \mathcal{S}_i^c    \right)
\\\nonumber &+   {\rm P}\left(\alpha^S_{m,II} >0, \forall m\in \mathcal{S},  \max \left\{ \gamma_m
, m\in \mathcal{S}   \right\} \leq \epsilon_s,   \right).
\end{align}

Note that  the probability $ {\rm P}\left(  \alpha^S_{1,II} = 0 \right)$   can be further bounded  as follows:
\begin{align}\label{eqdd13}
 {\rm P}\left(  \alpha^S_{1,II} = 0 \right) &\leq  {\rm P}\left(  g_1\leq  \frac{\epsilon_P}{\rho} \right) \\\nonumber &+{\rm P}\left(    h_1  \leq \epsilon_P \underset{i\neq 1}{\sum}h_i \alpha^P_{i,I} +\frac{\epsilon_P }{\rho}\right).
\end{align}
Denote the two  terms  at the right hand side of \eqref{eqdd13} by $Q_1$ and $Q_2$, respectively.  

Because the users' channels are independent and ideally complex Gaussian distributed,     $g_1=\left(\frac{[(\mathbf{G}^H\mathbf{G})^{-1}]_{i,i}}{M}\right)^{-1}$ follows the inverse-Wishart distribution, i.e.,  \cite{861781,7434594}
\begin{align}
f_{[(\mathbf{G}^H\mathbf{G})^{-1}]_{i,i}}(x) =&W^{-1}\left(I_M,N \right)\\\nonumber
=& \frac{1}{ \Gamma\left( N-M+1 \right)} x^{-(N-M+2)} e^{-\frac{1}{x}},
\end{align}
where $\Gamma(x)$  denotes the gamma function and $\gamma(n,x)$ denotes the incomplete gamma function \cite{GRADSHTEYN}. 
By using the distribution of $g_1$, 
$Q_1 $ can be evaluated as follows: 
\begin{align}
Q_1  =&\int_{\frac{\rho}{M\epsilon_P}}^{\infty}  \frac{1}{ \Gamma\left( N-M+1 \right)} x^{-(N-M+2)} e^{-\frac{1}{x}}
dx\\\nonumber =&  \frac{1}{ \Gamma\left( N-M+1 \right)} \gamma\left(N-M+1,\frac{M\epsilon_P}{\rho}\right). 
\end{align}
At high SNR, $Q_1$ can be approximated as follows:
\begin{align}
Q_1 =&  \frac{ \sum^{\infty}_{n=0}\frac{(-1)^n}{n!(N-M+1+n)}\left(\frac{M\epsilon_P}{\rho}\right)^{N-M+1+n}}{ \Gamma\left( N-M+1 \right)}
\\\nonumber
\approx&    \frac{1}{ (N-M+1)!}\left(\frac{M\epsilon_P}{\rho}\right)^{N-M+1}\rightarrow 0,
\end{align}
for $\rho\rightarrow \infty$.

Recall that $\alpha_{i,I}^P=\min\left\{1, \frac{\epsilon_P}{\rho g_i} \right\}$.  By using the definition $\mathcal{B}_1=\{2, \cdots, M\}$, $Q_2$ can be evaluated as follows:
\begin{align}\label{q2dz}
Q_2=&\underset{\mathcal{B}\subseteq \mathcal{B}_1}{\sum}{\rm P}\left(    h_1  \leq \epsilon_P \underset{i\neq 1}{\sum}h_i \alpha^P_{i,I} +\frac{\epsilon_P }{\rho},\right.\\\nonumber &
\left. \alpha_{j,I}^P= 1 ,j \in\mathcal{B}^c,  \alpha_{k,I}^P= \frac{\epsilon_P}{\rho g_k}, k \in \mathcal{B}\right),
\end{align}
where $\mathcal{B}^c$ denotes the complementary set of $\mathcal{B}$. Therefore, the probability can be upper bounded as follows:
\begin{align}\label{eq13zk}
Q_2 \leq  &
\underset{\mathcal{B}\subset \mathcal{B}_1}{\sum}{\rm P}\left(   \alpha_{j,I}^P= 1 ,j \in\mathcal{B}^c\right)\\\nonumber 
&+ {\rm P}\left(    h_1  \leq \epsilon_P \underset{i\neq 1}{\sum}h_i \alpha^P_{i,I} +\frac{\epsilon_P }{\rho},   \alpha_{i,I}^P= \frac{\epsilon_P}{\rho g_i}, i \in \mathcal{B}_1\right). 
\end{align} 
Denote the last term in \eqref{eq13zk} by $Q_3$ which can be expressed as follows:
\begin{align}\nonumber 
Q_3 =&   {\rm P}\left(  \xi_1 \leq \frac{\epsilon_P }{\rho},    g_i\geq  \frac{\epsilon_P}{\rho }, i \in \mathcal{B}_1\right) 
\leq {\rm P}\left( \xi_1 \leq \frac{\epsilon_P }{\rho} \right),
\end{align}
 where  $\xi_m=  \frac{ h_m}{\underset{i\neq m}{\sum}h_i \frac{\epsilon_P}{  g_i}+1} $.  Denote the CDF of $\xi_m$ by $F_{\xi_m}(x)$. Therefore, at high SNR, $Q_3$ can be approximated as follows: 
\begin{align}
Q_3  
\leq&   F_{\xi_1}\left( \frac{\epsilon_P }{\rho} \right)\rightarrow F_{\xi_1}\left(0\right)=0,
\end{align}
since $F_{\xi_m}(x)$ is not related to $\rho$. It is straightforward to show that ${\rm P}\left(   \alpha_{j,I}^P= 1 ,j \in\mathcal{B}^c\right)\rightarrow 0$ for $\rho \rightarrow \infty$, since  ${\rm P}\left(   \alpha_{j,I}^P= 1 ,j \in\mathcal{B}^c\right)={\rm P}\left(   \frac{\epsilon_P}{\rho g_j} \geq 1 ,j \in\mathcal{B}^c\right)$. Therefore, both the terms in \eqref{eq13zk}  goes to zero at high SNR, which means that $Q_2\rightarrow 0$ for $\rho \rightarrow \infty$. Since  $Q_i\rightarrow 0$, $i\in\{1,2\}$, $ {\rm P}\left(  \alpha^S_{1,II} = 0 \right)\rightarrow 0$ for $\rho\rightarrow \infty$. Similarly, one can also prove that $ {\rm P}\left( \alpha^S_{j,II} =0, \forall j\in \mathcal{S}_i^c    \right)$ in \eqref{eq7xx} goes to zero for  $\rho\rightarrow \infty$.

Denote the last term in \eqref{eq7xx} by $Q_4$. Similar to $Q_2$ in \eqref{eq13zk}, $Q_4$  can be   upper bounded as follows:
\begin{align}
&Q_4\leq  
\underset{\mathcal{B}\subset \mathcal{B}_1}{\sum}{\rm P}\left(   \alpha_{j,I}^P= 1 ,j \in\mathcal{B}^c\right)\\\nonumber 
&+ {\rm P}\left(    \max \left\{   \xi_m\alpha^S_{m,II}
, m\in \mathcal{S}   \right\} \leq \frac{\epsilon_s}{\rho} , \alpha_{i,I}^P= \frac{\epsilon_P}{\rho g_i}, i \in \mathcal{B}_1 \right) .
\end{align}
By using the two choices for $\alpha^S_{m,II}$ shown in \eqref{alphaII}, $Q_4$ can be further upper bounded as follows: 
\begin{align}
Q_4\leq  &\nonumber
\underset{\mathcal{B}\subset \mathcal{B}_1}{\sum}{\rm P}\left(   \alpha_{j,I}^P= 1 ,j \in\mathcal{B}^c\right) + {\rm P}\left(    \xi_1\alpha^S_{1,II} 
\leq \frac{\epsilon_s}{\rho} ,\alpha_{1,I}^P >0\right) \\\nonumber
=  &
\underset{\mathcal{B}\subset \mathcal{B}_1}{\sum}{\rm P}\left(   \alpha_{j,I}^P= 1 ,j \in\mathcal{B}^c\right)+ {\rm P}\left(     \frac{h_1\frac{ g_1- \frac{\epsilon_P}{\rho}}{(\epsilon_P +1)g_1} }{ \epsilon_P\underset{i\neq 1}{\sum}\frac{h_i}{ g_i}+1}
\leq \frac{\epsilon_s}{\rho}  \right)\\ \label{eq16d} 
& + {\rm P}\left(  \xi_1\leq \frac{\epsilon_P+\epsilon_s(1+\epsilon_P)}{\rho}  \right) . 
\end{align}
Following the steps  to bound $Q_1$ and $Q_2$, it is straightforward to show that the first and third terms in \eqref{eq16d} go to zero at high SNR. Furthermore,  note that ${\rm P}\left(     \frac{h_1\frac{ g_1- \frac{\epsilon_P}{\rho}}{(\epsilon_P +1)g_1} }{ \epsilon_P\underset{i\neq 1}{\sum}\frac{h_i}{ g_i}+1}
\leq \frac{\epsilon_s}{\rho}  \right) \rightarrow \tilde{F}_1(0)=0$ for $\rho\rightarrow \infty$, where $\tilde{F}_1(x)$ denotes the CDF of $  \frac{h_1\frac{ g_1 }{(\epsilon_P +1)g_1} }{ \epsilon_P\underset{i\neq 1}{\sum}\frac{h_i}{ g_i}+1}$ and  $\tilde{F}_1(x)$ is not related to $\rho$. Therefore, $Q_4\rightarrow 0$ for $\rho\rightarrow \infty$. Since all the three terms in \eqref{eq7xx} approach zero at high SNR, the proof is complete. 
\bibliographystyle{IEEEtran}
\bibliography{IEEEfull,trasfer}
  \end{document}